\documentclass[10pt,a4paper]{article}
\usepackage[utf8]{inputenc}
\usepackage{geometry}
\usepackage[titletoc,toc,title]{appendix}
\usepackage{graphicx}
\usepackage{amssymb}
\usepackage{epstopdf}
\usepackage{mathrsfs}
\usepackage{amsmath}
\usepackage{amsthm}
\begin{document}
\newtheorem{theoreme}{Theorem}
\newtheorem{ex}{Example}
\newtheorem{definition}{Definition}
\newtheorem{lemme}{Lemma}
\newtheorem{remarque}{Remark}
\newtheorem{exemple}{Example}
\newtheorem{proposition}{Proposition}
\newtheorem{corolaire}{Corollary}
\newtheorem{hyp}{Hypothesis}
\newtheorem*{rec}{Recurrence Hypothesis}
\newcommand\bel{>}
\newcommand\N{\mathbb{N}}
\newcommand\Z{\mathbb{Z}}
\newcommand\R{\mathbb{R}}
\newcommand\C{\mathbb{C}}
\newcommand\Sp{\mathbb{S}}
\newcommand\hp{\mathcal{T}^{d-1}X}

\title{Lower bounds for the number of nodal domains for sums of two distorted plane waves in non-positive curvature}
\author{Maxime Ingremeau}

\maketitle
\begin{abstract}
In this paper, we will consider generalised eigenfunctions of the Laplacian on some surfaces of infinite area. We will be interested in lower bounds on the number of nodal domains of such eigenfunctions which are included in a given bounded set.

We will first of all consider finite sums of plane waves, and give a criterion on the amplitudes and directions of propagation of these plane waves which guarantees an optimal lower bound, of the same order as Courant's upper bound.

As an application, we will obtain optimal lower bounds for the number of nodal domains of distorted plane waves on some families of surfaces of non-positive curvature.
\end{abstract}

\section{Introduction}
Let $(X,g)$ be a compact Riemannian manifold, and let us denote by $(\varphi_j)_{j\in \mathbb{N}}$ an orthonormal basis of $L^2(X)$ made of eigenfunctions of the Laplace-Beltrami operator:
\begin{equation*}
-\Delta_g \varphi_j = \lambda_j^2 \varphi_j.
\end{equation*}

The nodal domains of $\varphi_j$ are the connected components of $X\backslash \{\varphi_j^{-1}(0)\}$. Let us denote by $\mathcal{N}_j$ the number of nodal domains of $\varphi_j$. It is known since Courant (\cite{CouHil}) that we have
\begin{equation}\label{Courant}
\mathcal{N}_j = O(\lambda_j^2).
\end{equation}

This bound is in general not optimal. Indeed, we know since Stern that there exists some examples of spherical harmonics $\phi_j$ having only two nodal domains, while $\lambda_j\rightarrow \infty$ (\cite{stern1925bemerkungen}, see also \cite[Theorem 2.1.4]{HanLin}). 
However, it is thought that in a “generic” setting, the bound (\ref{Courant}) should be optimal.

On the two-dimensional torus, Buckley and Wigman (\cite{buckley2015number}), using ideas from Bourgain (\cite{bourgain2014toral}) were able to build many families $(\phi_j)$ of eigenvalues of $-\Delta$ which satisfied $\mathcal{N}_j \geq c \lambda_j^2$ for some $c > 0$, thus saturating the Courant bound. To do so, they were able to relate locally the nodal domains of trigonometric polynomials to the nodal domains of Random Gaussian Fields, and to use the powerful machinery developed by Nazarov and Sodin in this framework (\cite{nazarov2009number}, \cite{nazarov2015asymptotic}). Actually, Buckley and Wigman are able to show that $\mathcal{N}_j \sim c_0 \lambda_j^2$, where $c_0$ is a (hardly explicit) constant depending on the family $(\phi_j)$, known as the “Nazarov-Sodin constant”.

Gaussian Random Fields should be useful to describe nodal domains on manifolds which are more general that the torus. Indeed,
it is believed since the work of Berry \cite{berry1977regular} that generic eigenfunctions of $-\Delta$ on compact manifolds of negative curvature behave according to the so-called \emph{random wave model}, and hence their nodal domains should behave somewhat like those of Gaussian Random Fields. 

\subsubsection*{Nodal domains on manifolds of infinite volume}
In this paper, we will mainly be interested in eigenfunctions of the Laplacian on \emph{manifolds of infinite volume}, hence non-compact.
On such manifolds, there are no $L^2$-eigenfunctions, but in general, for any $h>0$\footnote{The parameter $h>0$ here corresponds to $\lambda_j^{-1}$ in the previous paragraph. We will therefore be considering the semi-classical limit $h\rightarrow 0$.} , there exists many solutions $\phi_h\in C^\infty(X)$ to the equation
$$-h^2\Delta_g \phi_h =  \phi_h.$$

If $\phi_h$ is such an eigenfunction, it will not be compactly supported, hence it may have infinitely many nodal domains.
However, if $\Omega\subset X$ is a bounded set, we may consider

\begin{equation}\label{supnod}
\begin{aligned}
N_{\subset\Omega}(\phi_h) &= \sharp\{ \text{ nodal domains of } \phi_h \text{ included in } \Omega\}.
\end{aligned}
\end{equation}

Note that, if $\Omega\subset \Omega'$, then $N_{\subset\Omega}(\phi_h)\leq N_{\subset\Omega'}(\phi_h)$.
Furthermore, for any $\Omega\subset X$ bounded with smooth boundary, there exists $C_{\Omega}$ such that 

\begin{equation}\label{CourantInfini}
N_{\subset\Omega}(\phi)\leq \frac{C_\Omega} {h^{d}}.
\end{equation}

To prove this bound, we may just use \cite[Lemme 16]{berard1982inegalites} (which generalizes results of \cite{pleijel1956remarks} and \cite{peetre1957generalization}), which gives us a constant $c_\Omega\bel 0$ such that for any solution of $(-h^2\Delta -1)f=0$, every nodal domain of $f$ included in $\Omega$ has a volume larger than $c_\Omega h^d$.

The estimate (\ref{CourantInfini}) can be seen as an analogue of (\ref{Courant}) on manifolds of infinite volume. Just as in the compact case, it is natural to wonder if a lower bound of the same order holds. We will give a positive answer to this question for certain eigenfunctions on some families of surfaces which are \emph{Euclidean near infinity}.

\subsubsection*{Distorted plane waves on Euclidean near infinity surfaces}
Consider a Riemannian surface $(X,g)$ such that there exists a bounded open set $X_0\subset X$ and $R_0\bel 0$ such that $(X\backslash X_0,g)$ and $(\mathbb{R}^2\backslash B(0,R_0),g_{eucl})$ are isometric (we shall say that such a surface is \emph{Euclidean near infinity}). 

The distorted plane waves on $X$ are a family of functions $E_h(x;\omega)$ with parameters $\omega\in\mathbb{S}^{1}$ (the direction of propagation of the incoming wave) and $h$ (a semiclassical parameter corresponding to the inverse of the square root of the energy) such that
\begin{equation}\label{eigen}
(-h^2\Delta_g-1) E_h(x;\omega,g)=0, 
\end{equation}
and which can be put in the form
\begin{equation}\label{jeanne}
E_h(x;\omega,g)= (1-\chi)e^{i x\cdot \omega /h} + E_{out}.
\end{equation}
Here, $\chi\in C_c^\infty$ is such that $\chi\equiv 1$ on $X_0$, and $E_{out}$ is \emph{outgoing} in the sense that it satisfies the \emph{Sommerfeld radiation condition}, were $|x|$ is the distance to any fixed point in $X$:
\begin{equation}\label{Sommerfeld}
\lim \limits_{|x|\rightarrow \infty} |x|^{(d-1)/2} \Big{(} \frac{\partial}{\partial |x|} - \frac{i }{h}\Big{)} E_h^1 = 0.
\end{equation}

It can be shown (cf. \cite[\S 2]{Mel} or \cite[\S 4]{Resonances}) that there is only one function $E_h(\cdot;\omega)$ such that (\ref{eigen}) is satisfied and which can be put in the form (\ref{jeanne}).
In the sequel, we will mainly be interested on the nodal domains of the \emph{sum of two distorted plane waves} with close enough directions of propagation.

To obtain results on the nodal domains of such  eigenfunctions, we need to make some assumptions on the classical dynamics of the geodesic flow on $(X,g)$.

\subsubsection*{Classical dynamics}
If $(X,g)$ is a Riemannian surface which is Euclidean near infinity. We denote by $(\Phi_g^t)_{t\in \R}: S^*X \mapsto S^*X$ the geodesic flow induced by the metric $g$.

The \emph{trapped set} for the metric $g$ is defined as
\begin{equation*}
K_g:=\{(x,\xi)\in S^*X; \Phi_g^t(x,\xi) \text{ remains in a bounded set for all } t\in \mathbb{R}\}.
\end{equation*}

In the sequel, we will always make the following two assumptions:
\begin{equation}\label{Hyperbolique}
K_g \text{ is a hyperbolic set for } \Phi_g^t.
\end{equation}
\begin{equation}\label{pression}
dim_{Haus} \big{(}K_g\big{)} < 2,
\end{equation}
where $dim_{Haus}$ denotes the Hausdorff dimension.

These two assumptions are stable by sufficiently small perturbations of the metric (for (\ref{Hyperbolique}), this is known as the \emph{Structural stability} of hyperbolic sets, cf. \cite[Chapter 17]{KH}). Note that if the sectional curvature is strictly negative in a neighbourhood of $\pi_X(K_g)$, where $\pi_X$ denotes the projection on the base manifold, then (\ref{Hyperbolique}) is automatically satisfied.

\subsubsection*{Generic perturbations of a metric}
Our result will concern distorted plane waves for a generic perturbation of a metric satisfying (\ref{Hyperbolique}) and (\ref{pression}). Let us define what we mean by generic.

Let $(X,g)$ be a Riemannian manifold, and $\Omega\subset X$ be a bounded open set. We denote by $\mathcal{G}_\Omega$ the set of metrics on $X$ which coincide with $g$ outside of $\Omega$. For any $k\geq 2$, the distance $\|g-g'\|_{C^k(\Omega)}$ between elements of $\mathcal{G}_\Omega$ is not intrinsic, since we define it using a coordinate chart. However, the topology this distance induces does not depend on the choice of coordinates.

Let $P(g')$ be a property which can be satisfied by a metric $g'$ on $X$. We shall say that $P$ is satisfied for a generic perturbation of $g$ in $\Omega$ if there exists an open neighbourhood $G_0$ of $g$ in $\mathcal{G}_\Omega$ such that the set of $\{g'\in G_0; P(g') \text{ is satisfied} \}$ is open and dense in $G_0$ for the $C^k(\Omega)$ topology.
\subsubsection*{Main theorem}

Our main theorem says that, for a generic perturbation of a metric satisfying (\ref{Hyperbolique}) and (\ref{pression}), the sum of the real parts of two distorted plane waves with close enough directions of propagation will have at least $c h^{-2}$ nodal domains in a given bounded set $\Omega$, for some $c>0$ depending on $\Omega$.

\begin{theoreme}\label{ivrogne}
Let $(X,g)$ be a Riemannian surface of non-positive curvature which is Euclidean near infinity, and which satisfies (\ref{Hyperbolique}) and (\ref{pression}). There exists $\epsilon\bel 0$ such that for any $\omega_0, \omega_1\in \Sp^1$ with $|\omega_0-\omega_1|<\epsilon$ and $\omega_0\neq \omega_1$, and for any non-empty open set $\Omega\subset X$, the following holds. For a generic $C^k(X_0)$ perturbation $g'$ of $g$, there exists a constant $c\bel 0$ and $h_0\bel 0$ such that for all $0<h<h_0$, we have
 the function
$$N_{\subset \Omega}\big{(}\Re(E_h(\cdot,\omega_0;g')) + \Re(E_h(\cdot, \omega_1;g'))\big{)}\geq ch^{-2}.$$ 
\end{theoreme}

The fact that we need to perturb the metric in a generic way is probably an artefact of the proof. However, it is not clear if we really need to have two distorted plane waves to produce $ch^{-2}$ nodal domains, or if a single distorted plane wave could do under some more stringent assumptions.

The cornerstone of the proof is Proposition \ref{amplitudesrandom}, which implies that the sum of three plane waves with random amplitudes will have a compact nodal domain with positive probability. Our proof, though elementary, works only in dimension 2, and we do not know if a similar result (with more plane waves) holds in higher dimension; if it did, Theorem \ref{ivrogne} would hold true in any dimension provided we replace the assumption on the Hausdorff dimension of trapped set by a topological pressure assumption as in \cite{Ing2}.

\subsubsection*{Idea of proof and organisation of the paper}
The proof will heavily rely on the results of \cite{Ing2}, which say that on manifolds of negative curvature with a condition on some topological pressure generalizing (\ref{pression}), distorted plane waves can be written locally as a sum of plane waves (see section \ref{rappel}).
 The phases of these plane waves are somehow "random", at least in a generic case, due to the chaotic dynamics induced by the negative curvature. However, the directions and amplitudes are perfectly deterministic, and the amplitudes decay exponentially. 
 
 The situation is therefore quite different from the framework of Gaussian random fields and from the Random Waves Model, and we are lead to study the nodal domains of a finite sum of plane waves with given amplitudes and direction of propagation, but random phases. More precisely, we look for criteria which guarantee that, with positive probability, such a function has at leat $cR^2$ nodal domains in a ball of radius $R$.

We will present such a criterion in section \ref{critere}. Though our study barely scratches the surface of the problem, the criterion we find is enough to obtain the desired result on sum of distorted plane waves, which we will prove in section \ref{preuveivrogne}.

\paragraph{Acknowledgement}
The author would like to thank Stéphane Nonnenmacher for supervising this project, and Igor Wigman for many explanation on his works. He would also like to thank Frédéric Naud for finding a mistake in the first version of the proof, and for useful discussion.

The author is partially supported by the Agence Nationale de la Recherche project GeRaSic
(ANR-13-BS01-0007-01).

\section{A criterion for a finite sum of plane waves to saturate the Courant bound}\label{critere}
\subsection{Definitions and statement of the criterion}
\subsubsection*{Stable nodal domains}
In the sequel, we will want to perturb slightly the functions we consider, so we have to give a definition of \emph{stable nodal domains}, which will not be affected by such perturbations.
\begin{definition}
Let $\Omega\subset \R^2$, and $f\in C^0(\Omega)$. Let $N\in\N$, $x_1,...,x_N\in \R^2$ and $\epsilon\bel 0$. We shall say that $x_1,...,x_N$ belong to different $\epsilon$-stable compact nodal domains of $f$ if for all $g\in C^0(\R^2)$ such that $\|g\|_{C^0} \leq \epsilon$, and for all $i,j=1,...,N$, $x_i$ belongs to a compact connected component of $\{x\in \Omega; f+g\neq 0\}$, and if $x_i$ and $x_j$ do not belong to the same connected component of $\{x\in \Omega; f+g\neq 0\}$.

If this is true for some choice of $x_1,...,x_N$, we shall say that \emph{$f$ has at least $N$ $\epsilon$-stable compact nodal domains.} We shall say that $f$ has $N$ $\epsilon$-stable compact connected components if $f$ has at least $N$ $\epsilon$-stable compact connected components, but $f$ does not have at least $N+1$ $\epsilon$-stable compact connected components.
\end{definition}

If $f\in C^0(\R^2)$, we shall write
\begin{equation}\label{defdomainenodal}
N_{f,\epsilon}(R)= \sharp \big{\{} \text{ Compact, } \epsilon-\text{stable nodal domains of } f \text { included in }B(0,R)\big{\}}.
\end{equation}

Note that $N_{f,\epsilon}$ is a non decreasing function.

In the sequel, we will be interested in compact nodal domains of a function of the form
 \begin{equation}\label{sommeondesplanes}
\sum_{i\in I} a_i \cos(k_i\cdot x + \theta_i).
\end{equation}

Theorem \ref{phasesrandom} below gives us a lower bound on the number of nodal domains of such a function, under some hypotheses on the direction $k_i$ and on the amplitudes $a_i$, which we shall now describe.

\subsubsection*{$\epsilon$-independence}
\begin{definition}
Let $k_1,...,k_n\in \R^2$, and let $\epsilon, T\bel 0$. We shall say that $k_1,...,k_n$ are $(\epsilon,T)$-independent if there exists $u\in\Sp^{1}$ such that for all $\theta,\theta'\in \mathbb{T}^n$, there exists $t(\theta,\theta')\in [0,T]$ such that
\begin{equation}\label{preskindep}
\Big{(}\theta + t (k_1\cdot u, ..., k_n\cdot u)\Big{)} \mod 1 \in B(\theta',\epsilon).
\end{equation}
We will sometimes say that $k_1,...,k_n$ are $\epsilon$-independent if there exists $T\bel 0$ such that $k_1,...,k_n$ are $(\epsilon,T)$-independent.
\end{definition}

Note that if a family $\mathbf{k}$ of vectors is $(\epsilon,T)$-independent, any non-empty subfamily of $\mathbf{k}$ is also $(\epsilon,T)$-independent.

For any $\epsilon>0$ and $n\in \N$, there exists $c(\epsilon)>0$ such that for any family of vectors $\mathbf{k}=(k_1,...,k_n)\in (\R^2)^n$, the family $\mathbf{k}$ is  $\epsilon$-independent if and only if there exists a $u\in \Sp^1$ such that
\begin{equation}\label{critereindep2}
\forall p_1,...,p_n\in \mathbb{Z}, \Big{(}\sum_i p_i k_i\cdot u=0\Big{)} \Rightarrow \Big{(} \forall i, |p_i|=0 \Big{)} \text{ ou } \Big{(} \exists i, |k_i|\geq c(\epsilon) \Big{)}.
\end{equation}

We refer the reader to \cite[\S 4]{berti2003drift} for a proof of this fact, and for a bound on $c(\epsilon)$.

By contraposition of (\ref{critereindep2}), the set of vectors which are not  $\epsilon$-independent is a union of a finite number of kernels of non-zero linear forms. Therefore, an application of Baire's Theorem gives us the following remark. 
\begin{remarque}\label{generik}
For any  $k_1,...,k_N\in \R^2$, for any $\epsilon,\delta\bel 0$, the set of $(k'_1,...,k'_N)\in (\R^{2})^N$ such that $(k_1+k'_1,...,k_N+k'_N)$ is $\epsilon$-independent and $|k'_i|\leq \delta$ for all $i=1,...,N$ is open and dense in $B(0,\delta)\subset \R^{2N}$.

Furthermore, if the family $(k_1,...,k_{n'})$ is $\epsilon$-independent for some $n'<N$, then the set of $(k'_{n'+1},...,k'_N)\in \R^{2(N-n')}$ such that $(k_1,...,k_{n'}, k_{n'+1}+ k'_{n'+1},...,k_N+k'_N)$ is $\epsilon$-independent and $|k'_i|\leq \delta$ for all $i=n'+1,...,N$ is open and dense in $B(0,\delta)\subset \R^{2(N-n')}$.
\end{remarque}

\subsubsection*{$\epsilon$-non-domination}
We shall ask that within the amplitudes $a_i$, there is not a subfamily of amplitudes which dominates all the others, in the sense of the following definition.
\begin{definition}
Let $\epsilon\bel 0$ and let $(a_i)_{i\in I}$ be a finite or countable family of real numbers. We shall say that $(a_i)_{i\in I}$ is \emph{$\epsilon$-non-dominated} if there exists $(u_i)_{i\in I}\in \{-1,1\}^{|I|}$ such that
\begin{equation*}
\Big{|} \sum_{i\in I} u_i a_i \Big{|} \leq \epsilon.
\end{equation*}
\end{definition}

For example, it is a standard exercise to show that if $I=\N$ and $a_i \longrightarrow 0$ but $\sum_{i\in \N} |a_i| = +\infty$, then $(a_i)$ is $\epsilon$-non-dominated for all $\epsilon\bel 0$.

If the $(a_i)_{i\in I}$ can be regrouped by pairs  $a_i, a_{i'}$ with $|a_i|=|a_{i'}|$, then the family will be $\epsilon$-non-dominated for any $\epsilon>0$. We will always be in this situation in section \ref{preuveivrogne}.

\subsubsection*{Statement of the criterion}

Let $\mathbf{k}=(k_i)_{i\in I}$ be family of vectors of $\Sp^{1}\subset \R^2$ indexed by a finite set $I$, and let $\mathbf{a}=(a_i)_{i\in I}$ be a set of positive real numbers indexed by $I$, such that $\sum_{i\in I} |a_i|^2=1$. We define the measure
\begin{equation*}
\mu_{\mathbf{k},\mathbf{a}} = \sum_{i\in I} |a_i|^2(\delta_{k_i}+\delta_{-k_i}),
\end{equation*}
which is a probability measure on $\Sp^{1}$, symmetric with respect to the origin.

If $\boldsymbol{\theta}=(\theta_i)_{i\in I}$ is a family of real numbers, we set
\begin{equation*}
f_{\mathbf{a},\mathbf{k},\boldsymbol{\theta}}(x) := \sum_{i\in I} a_i \cos(k_i\cdot x + \theta_i)
\end{equation*}

Recall that the quantity $N_{f,\epsilon}(r)$ has been defined in (\ref{defdomainenodal}).

\begin{theoreme}\label{phasesrandom}
Let $I$, $\mathbf{k}$, $\mathbf{a}$ and $\boldsymbol{\theta}$ be as above. Suppose that the measures $\mu_{\mathbf{k},\mathbf{a}}$ on $\Sp^{1}$ has at least $6$ points in its support.

Then there exists strictly positive constants $\mathcal{R}_0$, $\epsilon_0, \epsilon_1,\epsilon_2,\epsilon_3$ and $c$ depending only on $\sup_{i\in I} a_i$ and on the $6$ points in the support of $\mu$ and on their masses, such that the following holds.

Suppose that the vectors $(k_i)_{i\in I}$ are $(\epsilon_0,T_0)$-independent. Suppose furthermore that there exists a disjoint partition $\Sp^{1}= \bigsqcup_{l=1}^{L} S_l$ into sets of diameters all smaller that $\epsilon_1$, such that for all $l=1,...,L$, the set $\{a_i; i\in I ~~ \text{ and } k_i\in S_l \}$ is $\epsilon_2$ non-dominated.

Then for all $r\geq \mathcal{R}_0$, we have
\begin{equation*}
N_{f_{\mathbf{a},\mathbf{k},\boldsymbol{\theta}},\epsilon_3/2}(r) \geq c r^2.
\end{equation*}
\end{theoreme}

\begin{remarque}\label{stableperturb}
This result is stable by small perturbations of $\mathbf{a}$ and $\mathbf{k}$ in the following sense. Suppose that $I$, $\mathbf{k}$, $\mathbf{a}$ and $\boldsymbol{\theta}$ satisfy the hypotheses of the theorem. Then there exists $\epsilon_4>0$ such that, if $\mathbf{a}'$, $\mathbf{k}'$ and $\boldsymbol{\theta}'$ are such that $|\mathbf{a}- \mathbf{a}'|<\epsilon_4$ and $|\mathbf{k}'-\mathbf{k}|<\epsilon_4$, then
\begin{equation*}
N_{f_{\mathbf{a}',\mathbf{k}',\boldsymbol{\theta}}',\epsilon_3}(r) \geq \frac{c}{2} r^2.
\end{equation*}
\end{remarque}

\subsubsection*{An application to the torus}
The aim of this paragraph is to explain how Theorem \ref{ivrogne} can be used to find a lower bound on the number of nodal domains of some families of eigenfunctions on the torus.

These families will somehow be exceptional, since they are supported on a number of Fourier modes which does not depend on the frequency. This is hence very different from the framework of \cite{buckley2015number}, where the authors consider eigenfunctions which are supported on a large number of Fourier modes. It would be interesting to obtain a theorem which could describe the number of nodal domains in a larger framework containing these two situations.

Take $n\geq 3$, and fix any $k_1,...,k_n\in \Sp^1$ such that $k_i\neq \pm k_j$ for all $i,j\in \{1,...,n\}$. For any $\epsilon>0$ and $i=1,...,n$, we may find $k_i^\epsilon\in \Z^2$ and $\hat{k}_i^\epsilon$ such that
\begin{itemize}
\item
 $|k_i^\epsilon|=|\hat{k}_j^\epsilon|$ for any $i,j \in \{1,...,n\}$
 \item
  $\Big{|} k_i - \frac{k_i^\epsilon}{|k_i^\epsilon|}\Big{|}<\epsilon$ and $\Big{|} k_i - \frac{\hat{k}_i^\epsilon}{|\hat{k}_i^\epsilon|}\Big{|}<\epsilon$ for any $i=1,...,n$.
\item
The family $\bigcup_{i=1}^n \{k_i^\epsilon, \hat{k}^\epsilon_i\}$ is $\epsilon$-independent.
\end{itemize}
To obtain the last point, we simply made use of Remark \ref{generik}.

Take any sequence of amplitudes $a_1,...,a_n\in \R$ with $a_1, a_2,a_3 \neq 0$, and any sequence of real numbers $\theta_1,...,\theta_n, \hat{\theta}_1,...\hat{\theta}_n$.

The function $f^\epsilon \in C^\infty(\R^2)$ defined by
$$f^\epsilon(x) := \sum_{i=1}^n a_i \Big{(} \cos ( k^\epsilon_i\cdot x +\theta_i) +  \cos ( \hat{k}^\epsilon_i\cdot x +\hat{\theta}_i) \Big{)} $$
will then satisfy all the assumptions of Theorem \ref{phasesrandom}, provided that $\epsilon$ is taken small enough. Therefore, if $\epsilon$ has been taken small enough, we may find for any $\eta>0$ a constant $c>0$ such that 
\begin{equation}\label{extore}
N_{f^\epsilon,\eta(r)}\geq \frac{c}r^2.
\end{equation}

Now, for any $p\in \N$, the function
$$\phi_p(x):= f^\epsilon(px) = \sum_{i=1}^n a_i \Big{(} \cos ( pk^\epsilon_i\cdot x +\theta_i) +  \cos ( p\hat{k}^\epsilon_i\cdot x +\hat{\theta}_i) \Big{)} $$
defines a function on $\mathbb{T}^2$, which satisfies
$$-\Delta \phi_p= \lambda_p^2 \phi^p,$$
where $\lambda_p= p |k_1^\epsilon|$.

The bound (\ref{extore}) allows us to find a constant $c'>0$ such that
\begin{equation*}
N_{\phi_p} \geq c' \lambda_p^2.
\end{equation*}

\subsection{Proof of Theorem \ref{phasesrandom}}

The proof relies mainly on the following proposition, which we shall prove in the next subsection.
\begin{proposition}\label{amplitudesrandom}
There exists $\epsilon_5\bel 0$ such that the following holds.
Let $k_1,k_2,k_{3}\in \Sp^{1}$ be $(\epsilon_5,T)$-independent. Then there exists $R_0,\epsilon_4\bel 0$ such that for each $N\in \N$ and any $k'_1,...,k'_N\in \Sp^{1}$, there exist an open set $\Omega\subset \R^{3+N}$ such that for all $(a_1,a_2,a_{3},a'_1,...,a'_N)\in \Omega$, and $(\phi_1,\phi_2,\phi_{3},\phi'_1,...,\phi'_N)\in \R^{3+N}$, the function 
\begin{equation}\label{deffonction}
f(x):= \sum_{j=1}^{3} a_j \cos(k_j\cdot x+\phi_j)+\sum_{j=1}^N a'_j \cos(k'_j\cdot x+\phi'_j)
\end{equation}
 has an $\epsilon_4$-stable compact connected nodal domain in $B(0,T+R_0)$.
\end{proposition}

\begin{remarque}
This proposition implies that, if $\mu$ is a symetric measure on $\mathbb{S}^1$ with at least 6 points in its support, then the Nazarov-Sodin constant of $\mu$, as defined in \cite{kurlberg2015non} is strictly positive.
\end{remarque}

\begin{remarque}
The set $\Omega\subset \R^{3+N}$ given by the proposition is almost conical, in the following sense. If $(a_1,a_2,a_{3},a'_1,...,a'_N)\in \Omega$, then if $\lambda>0$, the function $$\sum_{j=1}^{3} \lambda a_j \cos(k_j\cdot x+\phi_j)+\sum_{j=1}^N \lambda a'_j \cos(k'_j\cdot x+\phi'_j)$$ has a compact nodal domain which is $\lambda \epsilon_4$-stable and included in $B(0,T+R_0)$.
\end{remarque}

Let us explain in an informal way the idea of the proof of Theorem  \ref{phasesrandom} from proposition \ref{amplitudesrandom}.

We want to consider the function $f_{\mathbf{a},\mathbf{k},\boldsymbol{\theta}}(x+y)=f_x(y)= \sum a_i \cos (k_i \cdot x + k_i \cdot y + \theta_i)$, by seeing $y$ as a variable, and $x$ as a parameter. To show that $f$ has at least $cr^2$ nodal domains in $B(O,r)$, we will show that for every point $x_0$, there exists a parameter $x$ close to $x_0$, such that $f_x$ has at leat a compact nodal domain. By covering $B(O,r)$ by $c r^2$ balls centred around different $x_0$ for some $c>0$, we will obtain the result.

Proposition \ref{amplitudesrandom} roughly says that if we consider a sum of plane waves with independent random amplitudes, we will have a compact nodal domain with probability $>0$.

A priori, we do not have random amplitudes here, by the hypothesis of $\epsilon$-independence between the directions of propagation $k_i$ roughly tells us that we can see the $k_i\cdot x$ as random phases. To go from random phases $\phi_i$ to random amplitudes, we want to use the following trick:
\begin{equation}\label{astuce}
\cos (k_i \cdot y  + \phi_i)+\cos (k_{i'} \cdot  y + \phi_{i'}) = 2 \cos \Big{(}\frac{\phi_i-\phi_{i'}}{2}\Big{)}\cos \Big{(}\frac{k_i+k_{i'}}{2}\cdot y + \frac{\phi_i+\phi_{i'}}{2}\Big{)}.
\end{equation}

The factor $2 \cos (\frac{\phi_i-\phi_{i'}}{2})$ can then be seen as a random amplitude.
Equation (\ref{astuce}) hence allows us to go from a sum of two plane waves with independent random phases and \emph{having the same amplitude} to a plane wave with a random amplitude.

To apply this trick, and put the function $f_x(y)$ in the framework of Proposition \ref{amplitudesrandom}, it is therefore essential that the amplitudes which we consider are two by two equals. It is the hypothesis of $\epsilon$-non-domination which will ensure us that we are almost in this situation, and which will allow us to prove the theorem. 

\begin{proof}[Proof that Proposition \ref{amplitudesrandom} implies Theorem   \ref{phasesrandom}]
Let $\epsilon_1\bel 0$, and consider a disjoint partition of $\Sp^{d-1}= \bigsqcup_{l=1}^{L} S_l$ into sets of diameters all smaller that $\epsilon_1$. Let us denote by $I_l\subset I$ the subset of indices such that $k_i\in S_l$. For each $l$, \emph{we fix a $i_l$ such that $k_{i_l}\in S_l$}.

By assumption, on $\mu=\mu_{\mathbf{k},\mathbf{a}}$, by possibly taking $\epsilon_1$ smaller, we may suppose that there exists $3$ sets $S_{l_1},S_{l_2},S_{l_{3}}$ such that $S_{l_j}\cap (- S_{l_j'})= \emptyset$ for all $j,j'\in \{1,2,3\}$, and such that we have 
\begin{equation}\label{nonvide}
\mu(S_{l_j}) > 0,
\end{equation}
 for some constant $c_j\bel 0$. 

Take $x,y\in \R^2$. We  have 
\[
f(x+y) = \sum_{i\in I} a_i \cos\big{(}k_i\cdot y + \theta_i(x)\big{)},
\]
where $$\theta_i(x) := k_i\cdot x+\theta_i.$$

For each $l=1,...,L$, let us write
\begin{equation*}
f^l(x):=\sum_{i\in I_l} a_i \cos(k_i\cdot x + \theta_i).
\end{equation*}

Hence, if $x,y\in \R^2$, we  have 
\[
\begin{aligned}
f^l(x+y) = \sum_{i\in I_l} a_i \cos\big{(}k_i\cdot y + \theta_i(x)\big{)}
&=\sum_{i\in I_l} a_i \big{(}\cos(k_{i_l}\cdot y + \theta_i(x))+ O(|y|\epsilon_1)\big{)}.
\end{aligned}
\]

\paragraph{Using the non-domination}
Suppose now that the set $\{a_i; i\in I_l\}$ is $\epsilon_2$-non-dominated for some $\epsilon_2\bel 0$.

We may then find a partition of $I_l$ into two subsets $J_l$ and ${J'}_l$ such that
\begin{equation}\label{preskegal}
\sum_{i\in J_l} a_i = \sum_{i\in {J'}_l} a_i + r_l,
\end{equation}
where $|r_l|< \epsilon_2$.

\begin{lemme}
Equation (\ref{preskegal}) implies that it is possible to build $p_i\in \N$ for each $i\in I_l$ and weights $t_1^i,...,t_{p_i}^i$ such that the following holds, where we write $\tilde{J}_l:= \{(i,j); i\in J_1  ~\text{and } j\leq p_i \}$, and $\tilde{J'}_l:= \{(i',j'); i'\in {J'}_1  ~\text{and } j'\leq p_{i'} \}$.

\begin{itemize}
\item For every $i\in I_l$, $\sum_{j=1}^{p_i} t_j^i=1$.
\item There exists a bijection $\tau: (i,j)\mapsto (i'(i,j),j'(i,j))$ between $\tilde{J}_l$ and $\tilde{J'}_l$ such that
\begin{equation*}
t_j^i a_i = t_{j'(i,j)}^{i'(i,j)} a_{i'(i,j)} + r_j^i,
\end{equation*}
where
\begin{equation*}
\sum_{(i,j)\in \tilde{J}_l} | r_j^i| \leq r_l< \epsilon_2.
\end{equation*}
\item The set $\tilde{J}_l$ has a cardinal lower or equal than $|I_l|$.
\end{itemize}
\end{lemme}

This lemma, a bit technical to state, simply says that it is possible to break the right-hand side and the left-hand side of (\ref{preskegal}) into small pieces, so that there are not more than $|I_l|$ pieces on the right and on the left, and that to each piece on the left corresponds a piece on the right which has almost the same amplitude.

\begin{proof}
The proof is done by recurrence on $|I_l|$. If the set has cardinal 2, the result is obvious by taking $p_i=1$, $t^i_1=1$ for the two elements $i$.

Suppose that $|I_l|$ has a cardinal grater than two. There is at least one smaller element in the $a_i$, for $i\in I_l$, which we shall write $a_{i_0}$. We may suppose for instance, without loss of generality, that $i_0\in J_l$. Take any $i'_0\in {J'}_l$. (\ref{preskegal}) may be rewritten as
\begin{equation*}
\sum_{i\in J_1\backslash \{i_0\}} a_i = \sum_{i\in {J'}_l \backslash i'_0} a_i + \big{(}1- \frac{a_{i_0}}{a_{i'_0}}\big{)} a_{i'_0} +r_l.
\end{equation*}
By applying the recurrence hypothesis to this new equation which contains one less term, we may deduce the lemma.
\end{proof}

We therefore have
\begin{equation*}
\begin{aligned}
f^l(x+y) &= \sum_{(i,j)\in \tilde{J}_l} t_j^i a_i\big{(}\cos(k_{i_l}\cdot y + \theta_i(x))+ O(|y|\epsilon_1)\big{)} \\
&+ \sum_{(i',j')\in {\tilde{J'}}_l} t_{j'}^{i'} a_{i'}\big{(}\cos(k_{{i'}_l}\cdot y + \theta_{i'}(x))+ O(|y|\epsilon_1)\big{)}\\
&= \sum_{(i,j)\in \tilde{J}_l} \Big{[} t_j^i a_i\big{(}\cos(k_{i_l}\cdot y + \theta_i(x))+ O(|y|\epsilon_1)\big{)} \\
&+ t_j^i a_i \big{(}\cos(k_{i_l}\cdot y + \theta_{i'(i,j)}(x))+ O(|y|\epsilon_1)\big{)}\Big{]}+ r_l(x,y)\\
&=  \sum_{(i,j)\in \tilde{J}_l} 2 t_j^i a_i\Big{[}\cos\big{(} \frac{\theta_i(x)-\theta_{i'(i,j)}(x)}{2}\big{)}\\
&\times \cos \big{(}k_{i_l}\cdot y+ \frac{\theta_i(x)+\theta_{i'(i,j)}(x)}{2}\big{)}+ O(|y|\epsilon_1)\Big{]}+ r_l(x,y)
\end{aligned}
\end{equation*}
where $|r_l(x,y)|< \epsilon_2$ for all $x,y\in \R^d$.

\subsubsection*{Turning independent vectors into independent amplitudes}
Since we assume that the vectors $(k_i)_{i\in I}$ are $(\epsilon_0,T_0)$-independent for some $T_0>0$, we have that for all $\psi_i\in \mathbb{T}^{|I|}$ and for all $x_0\in \R^2$, there exists $x\in B(x_0,T_0)$ such that for all $i\in I$, $|\theta_i(x)-\psi_i|\leq \epsilon_0$.

Since we have $|\tilde{J}_l|\leq |I_l|$, this means that the phases $\frac{\theta_i(x)-\theta_{i'(i,j)}(x)}{2}$ can $\epsilon$-approach any $\psi_i\in \mathbb{T}^{|I|}$ by moving $x$ in $B(x_0,T_0)$.

In particular, for all $x_0\in \R^2$ and for any sequence $\big{(}b_{(i,j)}\big{)}_{(i,j)\in\bigcup_l \tilde{J}_l}$ with $|b_{(i,j)}|\leq 3 |t^i_j a_i|$, we may find $x\in B(x_0,T_0)$ such that for all $(i,j)\in\bigcup_l\tilde{J}_l$, we have
\[\Big{|}b_{(i,j)}-2 t_j^i a_i\cos\big{(} \frac{\theta_i(x)-\theta_{i'(i,j)}(x)}{2}\big{)}\Big{|} \leq 2\epsilon_0 t_j^i a_i\leq C_0 \epsilon_0, \]
with $C_0:= \max\limits_{(i,j)\in\bigcup_l \tilde{J}_l} 2 t_j^i a_i$.

\subsubsection*{Applying Proposition \ref{amplitudesrandom}}
For each $x\in \R^d$, set 
\begin{equation*}
f_x(y)= \sum_{l=1}^L \sum_{(i,j)\in \tilde{J}_l} \Big{[}2 t_j^i a_i\cos\Big{(} \frac{\theta_i(x)-\theta_{i'(i,j)}(x)}{2}\Big{)}\Big{]} \cos \Big{(}k_{i_l}\cdot y+ \frac{\theta_i(x)+\theta_{i'(i,j)}(x)}{2}\Big{)}.
\end{equation*}

We want to apply Proposition \ref{amplitudesrandom} to the function $f_x$.

Thanks to (\ref{nonvide}), for each $S_{j_1},S_{j_2},S_{j_{3}}$, we may find $(i,j)\in \tilde{J}_{l_j}$ such that $t_j^ia_i\neq 0$. The $3$ terms containing these in the definition of $f_x$ will correspond to the $3$ first terms in (\ref{deffonction}), while the remaining terms in $f_x$ will correspond the remaining terms in (\ref{deffonction}).

We want to make sure that the amplitudes $\Big{[}2 t_j^i a_i\cos\big{(} \frac{\theta_i(x)-\theta_{i'(i,j)}(x)}{2}\big{)}\Big{]}$ fall in the open set $\Omega$ described in Proposition \ref{amplitudesrandom}. Thanks to the previous paragraph, we know that, if we take $\epsilon_0$ small enough, we can always find $x\in B(x_0,T_0)$ such that this is true. 

We obtain that the function $f_x$ has an $\epsilon_4$-stable nodal domain in $B(0,T+R_0)$.

Since we have
\begin{equation*}
\big{|}f_x(y)-f_{\mathbf{a},\mathbf{k},\boldsymbol{\theta}}(x+y)\big{|} \leq \epsilon_2 + O(|y|\epsilon_1),
\end{equation*}
we get that if we have $\epsilon_1$, $\epsilon_2$ and $\epsilon_3$ small enough, then $f_{\mathbf{a},\mathbf{k},\boldsymbol{\theta}}$ has an $\epsilon_3$ stable nodal domain in $B(x,T+R_0)$, and hence has an $\epsilon_3$-stable nodal domain in $B(x_0,T+R_0+T_0)$ for any $x_0$. We may then find $c\bel 0$ such that there are $c R^2$ disjoint balls of radius $T+R_0+T_0$ in $B(0,R)$ for $R$ large enough. Since each of these balls contains an $\epsilon_3$-stable domain for $f_{\mathbf{a},\mathbf{k},\boldsymbol{\theta}}$, the theorem follows.
\end{proof}

\subsection{Proof of Proposition \ref{amplitudesrandom}}\label{preuveampl}
Let $\mathbf{k}=(k_1,k_2,k_{3})\in \big{(}\Sp^{1}\big{)}^{3}$ be such that for any $1\leq ,j'j\leq 3$, $j\neq j'$, we have $k_j\neq \pm k_{j'}$.

For any family $\mathbf{a}= (a_1,a_2,a_{3})\in \R^{3}$, write
\begin{equation*}
g_{\mathbf{a}}(x) := \sum_{i=1}^{3} a_i\cos(k_i\cdot x).
\end{equation*}

The proof of Proposition \ref{amplitudesrandom} relies on the following lemma :
\begin{lemme}\label{triangle}
There exists $\epsilon_6, R_0\bel 0$ and an open set $\Omega_{\mathbf{k}}\subset \R^{3}$ such that for all $\mathbf{a}\in \Omega_{\mathbf{k}}$, 0 belongs to an $\epsilon_6$-stable compact nodal domain of $g_{\mathbf{a}}$, and this nodal domain is contained in $B(0,R_0)$. 
\end{lemme}

Note that the set $\Omega_{\mathbf{k}}$ is almost a cone, in the sense that if $\mathbf{a}\in \Omega_\mathbf{k}$ and if $\lambda\in \R\backslash\{0\}$, then zero belongs to a $(|\lambda|\epsilon_6)$-stable compact nodal domain of $g_{\lambda\mathbf{a}}$.

\begin{proof}[Proof that Lemma \ref{triangle} implies Proposition \ref{amplitudesrandom}]
Suppose that $\mathbf{k}=(k_1,k_2,k_{3})\in \big{(}\Sp^{1}\big{)}^{3}$ are $(\epsilon_5,T)$-independent, for some $\epsilon_5$ to be determined later.

Let $(\phi_1,\phi_2,\phi_{3})\in \R^{3}$. We may find $x\in B(0,T)$ such that for all $j\in \{1,2,3\}$, we have $|k_j\cdot x + \phi_j|\leq \epsilon_5$. We hence have for all $j\in \{1,2,31\}$ and for all $y\in \R^2$:
\[ \Big{|}\cos \big{(}k_j\cdot (x+y) + \phi_j\big{)} - \cos \big{(}k_j\cdot y\big{)}\Big{|}\leq C \epsilon_5, \]
for some universal constant $C\bel 0$.

In particular, if we take some coefficients $a_1,a_2,a_{3}\in \Omega_{\mathbf{k}}$ as in Lemma \ref{triangle}, and if $\epsilon_5$ is chosen small enough so that $C \epsilon_5 \sup_{i=1,2,3} |a_i| \leq \frac{\epsilon_6}{2}$, we see that $\sum_{j=1}^{3} a_j \cos(k_j\cdot x+\phi_j)$ has an $(\epsilon_6/2)$-stable compact nodal domain in $B(0,R_0+T)$.

Now, if $N\in \N$, we just have to impose that for all $j=1,...,N$, we have $|a'_j
|\leq \frac{\epsilon_6}{4N}$ to make sure that the function
\begin{equation*}
\sum_{j=1}^{3} a_j \cos(k_j\cdot x+\phi_j)+\sum_{j=1}^N a'_j \cos(k'_j\cdot x+\phi'_j)
\end{equation*}
 has an $\frac{\epsilon_6}{4}$-stable compact connected nodal domain in $B(0,T+R_0)$. This concludes the proof of the proposition.
\end{proof}

Before proving Lemma \ref{triangle}, let us give an informal sketch of the proof. Consider first the sum of two cosine $a_1\cos (k_1\cdot x)+ a_2 \cos (k_2\cdot x)$. It will never have a compact nodal domain as soon as $|a_1|\neq |a_2|$. However, if $|a_1|$ and $|a_2|$ are very close to each other, the nodal domains are very thin in certain places, as represented in Figure \ref{example}. By adding a third cosine in a precise way, it is possible to "obstruct these thin passages", thus building a compact nodal domain around the origin.

\begin{figure}
    \center
   \includegraphics[scale=0.6]{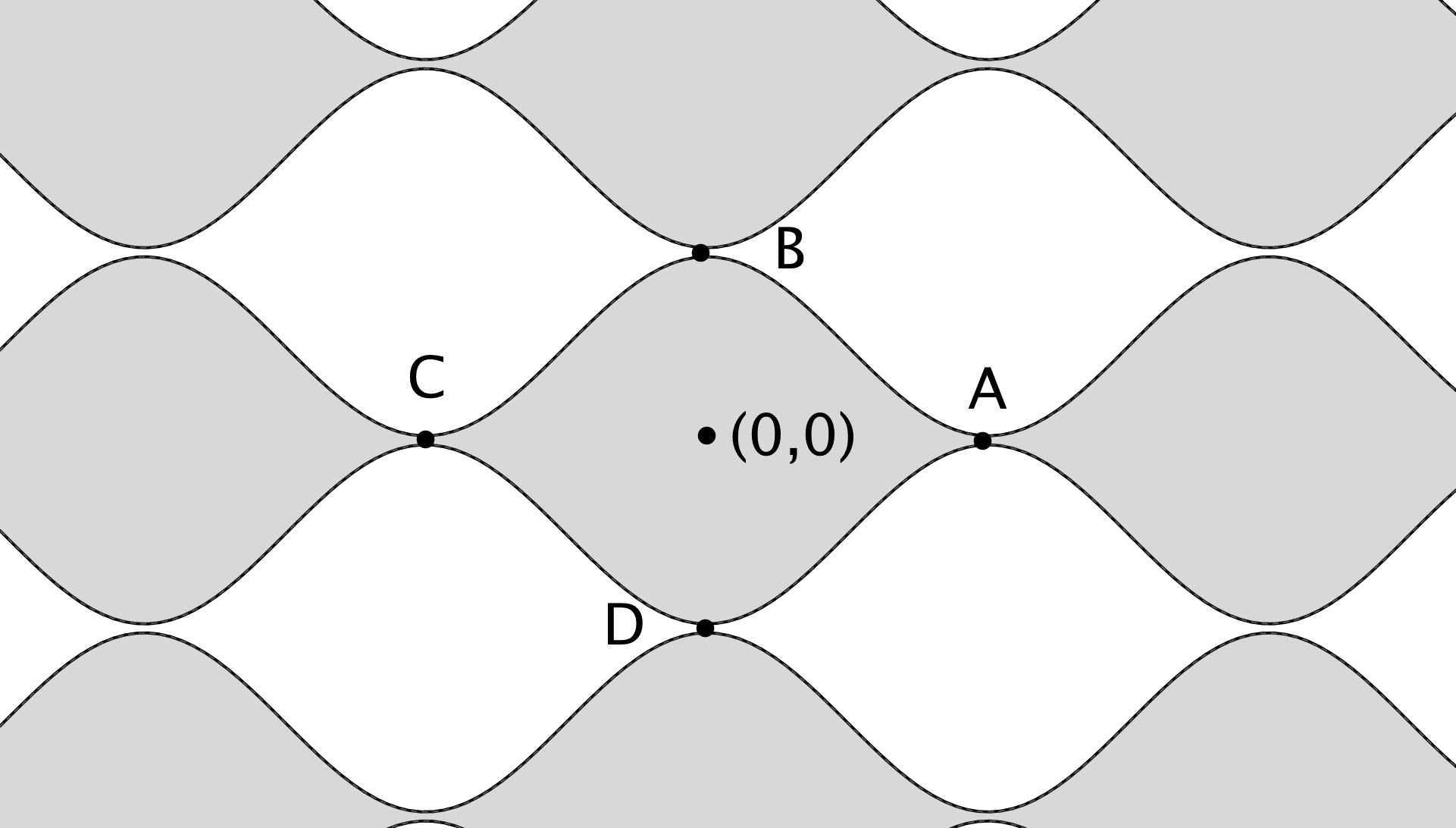}
    \caption{The sign of the function $f(x,y)= a_1\cos(x)+ a_2\cos(y)$, with $|a_2|$ slightly larger than $|a_1|$. $f$ is positive in the region in grey, and negative in the region in white. We want to add a third cosine which is positive in $A$ and $C$, and negative (or at least not too positive) in $B$ and $D$, so as to "close" the nodal domain containing $(0,0)$.}\label{example}
\end{figure}

\begin{proof}[Proof of Lemma \ref{triangle}]
We have by hypothesis three non zero real numbers $\lambda,\mu,\nu$ such that $\lambda k_1+ \mu k_2 + \nu k_3 =0$. Dividing by the coefficient with the greatest modulus and exchanging the vectors, we may assume that
\begin{equation*}
k_3 = \lambda' k_1+ \mu' k_2,
\end{equation*}
with $|\lambda'|\leq 1$, $|\mu'|\leq1$.

Furthermore, we must have
\begin{equation}\label{inegal}
(|\lambda'|- 1/2) \neq (1/2-|\mu'|).
\end{equation}
Indeed, if there were equality in (\ref{inegal}), then we would have $|\lambda' k_1| + |\mu' k_2|=|\lambda'|+|\mu'|=1=|k_3|=|\lambda' k_1 + \mu' k_2| $, which would imply that $k_1$ and $k_2$ are collinear.

In particular, we have $\cos(\mu'\pi)\neq - \cos(\lambda'\pi)$. Without loss of generality, we may thus suppose that

\begin{equation}\label{cosdif}
\text{If } \cos(\mu'\pi) \text { and } \cos(\lambda'\pi) \text{ do not have the same sign, then } |\cos(\lambda'\pi)| < |\cos(\mu'\pi)|.
\end{equation}

Without loss of generality, we will always suppose that $a_1\bel 0$.

\paragraph{Step 1 : understanding the sum of two cosine}

From now on, we will suppose that $$a_1-\varepsilon< a_2<a_1,$$ with $\varepsilon<<1$ to be determined later.

We shall write $S_1:= \{x\in \R^2; x\cdot k_1 = \pm \pi \text{ and } x\cdot k_2 \in [-\pi;\pi]\}$.
For $x\in S_1$, we have
$$g_{a_1,a_2,0}(x)= -a_1 + a_2 \cos (k_2 \cdot x) \leq -\varepsilon <0.$$

Furthermore, we have for $x\in S_1$, 
\begin{equation}\label{zombie}
g_{a_1,a_2,0}(x) = -a_1 + a_2 (1-(k_2\cdot x)^2/2 + o((k_2\cdot x)^2) \leq -\varepsilon - (k_2\cdot x)^2 + o((k_2\cdot x)^2). 
\end{equation}

We deduce from this that for any $A\bel 0$, there exists $c_A \bel 0$ independent of $\varepsilon$ and a $\varepsilon_A\bel 0$ such that for all $0 < \varepsilon \leq \varepsilon_A$ and for all $x\in S_1$, we have :

\begin{equation}\label{a}
|x\cdot k_2| \geq c_A \sqrt{\varepsilon}\Longrightarrow g_{a_1,a_2,0}(x) < -A \varepsilon
\end{equation}

Next, we consider the set $S_2:=\{x\in \R^2 \text{ such that } x\cdot k_2 = \pm \pi \text{ and such that }x\cdot k_1 \in [-\pi;\pi]\}$.
If $x\in S_2$, we have
\begin{equation*}
g_{a_1,a_2,0}(x) = a_1 \cos (k_1\cdot x) - a_2\leq \varepsilon.
\end{equation*}

Furthermore, just as before, we see that for each $B\bel 0$, there exists a $c_B\bel 0$ independent of $\varepsilon$ and an $\varepsilon_B\bel 0$ such that for all $0<\varepsilon \leq \varepsilon_B$, and for all $x\in S_2$, we have

\begin{equation}\label{b}
|x\cdot k_1| \geq c_B \sqrt{\varepsilon}\Longrightarrow g_{a_1,a_2,0}(x) < -B \varepsilon.
\end{equation}

\paragraph{Step 2 : adding a third cosine}

We now consider the function $g_{a_1,a_2,a_3}$. As long as $|a_3| < 2|a_1|-\varepsilon$, we have $g_{a_1,a_2,a_3}(0)\bel 0$. Let us find conditions on $a_3$ which will guarantee that $g_{a_1,a_2,a_3} (x) < 0$ if $x\in S_1\cup S_2$, which will show that $g_{a_1,a_2,a_3}$ has a compact nodal domain.

Let $x\in S_1$ be such that $|x\cdot k_2|\leq \varepsilon^{1/4}$. Then we have $k_3\cdot x = \lambda' k_1 \cdot x + \mu' k_2 \cdot x = \lambda'\pi + O(\varepsilon^{1/4})$. Therefore, we have
\begin{equation}\label{c}
\begin{aligned}
|x\cdot k_2|&\leq \varepsilon^{1/4}\Longrightarrow \cos(k_3\cdot x)= \cos (|\lambda'|\pi) + O(\varepsilon^{1/4}).
\end{aligned}
\end{equation}

Similarly, for all $x\in S_2$, we have
\begin{equation}\label{d}
|x\cdot k_1|\leq \varepsilon^{1/4}\Longrightarrow \cos(k_3\cdot x) = \cos(\mu'\pi) +  O(\varepsilon^{1/4}).
\end{equation}

\textbf{Suppose first that $\cos(\lambda' \pi)$ and $\cos (\mu' \pi)$ have the same sign.}

This means that $|\lambda'|, |\mu'| \geq 1/2$, so that the sign of $\cos(\lambda' \pi)$ must be negative.
We may take $$a_3\in \Big{[}-\varepsilon- \frac{2\varepsilon}{\min(|\cos(\lambda' \pi)|,|\cos(\mu' \pi)|)}; - \frac{2\varepsilon}{|\cos(\lambda' \pi)|}\Big{]}.$$

If $a_3$ is chosen so, then we have $g_{a_1,a_2,a_3}(0)\bel 0$ as long as $\varepsilon$ is small enough.
Take $$A= B = 2+ 2/\min(|\cos(\lambda' \pi)|,|\cos(\mu' \pi)|),$$ and $c_A$, $c_B$ as above.
Since $|a_3| < A\varepsilon,B\varepsilon$, we see from (\ref{a}) that for all $x\in S_1$ such that $|x\cdot k_2|\geq c_A \sqrt{\varepsilon}$, we have $g_{a_1,a_2,a_3}(x) <0$. Similarly, from (\ref{b}), we have that for all $x\in S_2$ such that $|x\cdot k_1|\geq c_B \sqrt{\varepsilon}$, we have $g_{a_1,a_2,a_3}(x) <0$.

Now, if $x\in S_1$ is such that $|x\cdot k_2|\leq c_A \sqrt{\varepsilon}$, or if $x\in S_2$ is such that $|x\cdot k_2|\leq c_B \sqrt{\varepsilon}$  we have from (\ref{c}) and (\ref{d}) that $a_3 \cos (k_3\cdot x) \leq -2\varepsilon$. Therefore, $g_{a_1,a_2,a_3}(x) <0$ for all $x\in S_1\cup S_2$. Therefore, $g_{a_1,a_2,a_3}$ has a compact nodal domain which is $\epsilon_6$-stable for $\epsilon_6$ small enough, and which belongs to $B(0,R_0)$ for $R_0$ large enough.

All in all, we have shown that $g_{a_1,a_2,a_3}$ has an $\epsilon_6$-stable compact nodal domain in $B(0,R_0)$ for all $(a_1,a_2,a_3)$ such that $a_1-a_2\in (0,\varepsilon_0)$ and $a_3 \in (- (a_2-a_1)- (a_1-a_2)/c; - (a_1-a_2)/c)$ for some $c$ depending only on $k_1,k_2,k_3$. This is a non-empty open set, and the connected  which proves the lemma.

\textbf{Suppose that $\cos(\lambda' \pi)$ and $\cos (\mu' \pi)$ have opposite signs.}

Then (\ref{cosdif}) implies that $ |\cos(\lambda'\pi)| < |\cos(\mu'\pi)|$. In particular, we have $|\cos(\mu'\pi)|\neq 0$.

Take $$a_3 \in \Big{[} \frac{-sgn(\cos(\mu' \pi)) \varepsilon}{1/3|\cos (|\mu'|\pi))|+ 2/3 |\cos (\lambda \pi)|} ;  \frac{- sgn(\cos(\mu' \pi))\varepsilon}{2/3|\cos (|\mu'|\pi))|+ 1/3 |\cos (\lambda \pi)|}\Big{]}.$$

If $a_3$ is chosen so, then we have $g_{a_1,a_2,a_3}(0)\bel 0$ as long as $\varepsilon$ is small enough.
Take $$A= B = 1/|\cos(\lambda' \pi)|,$$ and $c_A$, $c_B$ as above.
Since $|a_3| < A\varepsilon,B\varepsilon$, we see from (\ref{a}) that for all $x\in S_1$ such that $|x\cdot k_2|\geq c_A \sqrt{\varepsilon}$, we have $g_{a_1,a_2,a_3}(x) <0$. Similarly, from (\ref{b}), we have that for all $x\in S_2$ such that $|x\cdot k_1|\geq c_B \sqrt{\varepsilon}$, we have $g_{a_1,a_2,a_3}(x) <0$.

Now, if $x\in S_1$ is such that $|x\cdot k_2|\leq c_A \sqrt{\varepsilon}$,
we have from (\ref{c}) that $$0< a_3 \cos (k_3\cdot x) \leq \frac{\varepsilon|\cos(|\lambda' \pi)|}{1/3|\cos (|\mu'|\pi))|+ 2/3 |\cos (\lambda' \pi)|} + o(\varepsilon).$$ Since $\frac{|\cos(|\lambda' \pi)|}{1/3|\cos (|\mu'|\pi))|+ 2/3 |\cos (\lambda' \pi)|}<1$, we deduce from (\ref{zombie}) that $g_{a_1,a_2,a_3} <0$ on $S_1$.

If $x\in S_2$ is such that $|x\cdot k_2|\leq c_B \sqrt{\varepsilon}$  we have from (\ref{d}) that $$|a_3 \cos (k_3\cdot x)| \geq \frac{\varepsilon|\cos(|\mu' \pi)|}{2/3|\cos (|\mu'|\pi))|+ 1/3 |\cos (\lambda \pi)|} + o(\varepsilon).$$ Since
$\frac{|\cos(|\mu' \pi)|}{2/3|\cos (|\mu'|\pi))|+ 1/3 |\cos (\lambda \pi)|}\bel 1$, we deduce from (\ref{d}) that $g_{a_1,a_2,a_3} < 0$ on $S_2$.

Hence, $g_{a_1,a_2,a_3}(x) <0$ for all $x\in S_1\cup S_2$. Therefore, $g_{a_1,a_2,a_3}$ has a compact nodal domain, which is $\epsilon_6$-stable for $\epsilon_6$ small enough, and contained in $B(0,R_0)$ for $R_0$ large enough.

All in all, we have shown that $g_{a_1,a_2,a_3}$ has a compact nodal set for all $(a_1,a_2,a_3)$ such that $a_1-a_2\in (0,\varepsilon_0)$ and $a_3 \in ((a_1-a_2)/c; - (a_1-a_2)/c')$ for some $c,c'$ depending only on $k_1,k_2,k_3$. This is a non-empty open set, which proves the lemma.
\end{proof}

\section{Proof of Theorem \ref{ivrogne}}\label{preuveivrogne}

Before proving Theorem \ref{ivrogne}, let us recall the main fact from \cite{Ing2} which we will use in the proof.

\subsection{Recall of the results of \cite{Ing2}}\label{rappel}
Let $\mathcal{O}\subset X$ be a bounded open set, and let $\chi\in C_c^\infty(X)$ be equal to $1$ on $\mathcal{O}$. The main result in \cite{Ing2} implies that we can write  
\begin{equation}\label{qad2} 
\chi E_h(x,\omega;g) = \sum_{n=0}^{M |\log h|} \sum_{\beta\in \mathcal{B}_{\chi,n}} a_{\beta}(x;\omega,g,h) e^{i \varphi_{\beta}(x;\omega,g)/h} + R_{h}.
\end{equation}

Here, $M\bel 0$, and $\mathcal{B}_{\chi,n}$ is a set whose cardinal grows exponentially with $n$.
The $a_{\beta}$ are smooth functions of $x,\omega$, and their derivatives are bounded independently of $h$. The $\varphi_{\beta}$ are smooth function defined in a neighbourhood of the support of
$a_{\beta}$. 
We have
 \begin{equation*}\|R_{h}\|_{C^0}=O(h).
\end{equation*}

Furthermore there exists $\mathcal{P}<0$\footnote{$\mathcal{P}$ is actually the topological pressure associated to half the unstable Jacobian of the flow on the trapped set (see \cite{Ing2} for more details). The fact that this number is negative is equivalent, in dimension 2,  to condition (\ref{pression}).} such that, for any $\ell\in \mathbb{N}$, $\epsilon>0$, there exists $C_{\ell,\epsilon}$
such that
\begin{equation}\label{sheriff3}
\sum_{\beta\in\mathcal{B}_{\chi,n}} \|a_{\beta}\|_{C^1} \leq C_{\ell,\epsilon}
e^{n(\mathcal{P}+\epsilon)}.
\end{equation}

It was shown in \cite[Corollary 2]{Ing2} that for any $x\in X$ $\omega\in \mathbb{S}^1$, and $n_0\in \N$ we have
\begin{equation}\label{nonnul}
\sum_{n\geq n_0} \sum_{\beta\in \mathcal{B}_{\chi,n}} |a_\beta(x;\omega,g)|\bel 0.
\end{equation}

To obtain (\ref{qad2}), the author built a well-chosen open cover of $S^*X$, denoted $(V_b)_{b\in B}$, with all the $V_b$ bounded except one. The set $\mathcal{B}_{\chi,n}$ is actually a set of words on the alphabet $B$, of length approximately $n$.

\subsubsection*{Interpretation of $\varphi_\beta$ in terms of classical dynamics}
For each $\omega\in \Sp^1$, define 
\begin{equation*}\Lambda_{\omega}= \{(x,\omega); x\in X\backslash X_0\}.
\end{equation*}

If $\beta= b_1,...,b_N$, define 
\begin{equation*}
\Phi^N_\beta (\Lambda_\omega)= V_{b_N}\cap\Phi^1 \big{(}V_{b_{N-1}}\cap \Phi^1\big{(}...V_{b_1}\cap\Phi^1(\Lambda_\omega)... \big{)\big{)}}.
\end{equation*}

It was shown in \cite{Ing2} that
\begin{equation}\label{liendynamique}
\Phi^N_\beta(\Lambda_\omega) = \{(x,\nabla_x \varphi_\beta(x;\omega)); x\in O_\beta\},
\end{equation}
for some set $O_\beta\subset \pi_X(V_{b_N})$.
Therefore, $\nabla_x\varphi_x(x;\omega,g)$ is the direction of the unique trajectory coming from $\Lambda_\omega$ which was in $V_{b_k}$ at time $k$, and which is above $x$ at time $N$.

\begin{remarque}\label{Plusieurscoefs}
As explained in \cite{Ing2} (this is, for instance, a consequence of Corollary 4), if $K\neq \emptyset$, then for any $x\in X$ and any $\omega\in \Sp^1$, the vectors $\nabla_x \varphi (x,\omega,g)$ are different for different values of $\beta$.
\end{remarque}
\subsection{From Theorem \ref{phasesrandom} to Theorem \ref{ivrogne}}

\begin{proof}
Let us fix $\chi\in C_c^\infty(X)$ be equal to one on $\Omega$.
From now on, let us fix $x_0\in \Omega$, and consider a local chart $\psi$ from an neighbourhood of the origin in $\R^2$ to an open neighbourhood of $x_0$ included in $\Omega$. For all $\eta>0$, the results of section \ref{rappel} give us a $M_\eta>0$ such that for all $h>0$ small enough, we have for all $x\in B(0,h^{-1/3})$
\[\begin{aligned} E_h(\psi(h x);\omega)&= \sum_{n=0}^{M_\eta} \sum_{\beta\in \mathcal{B}_{\mathcal{K},n}} \big{(}a^0_{\beta}(x_0;\omega)+O(|x| h)\big{)} e^{\frac{i}{h} \varphi_{\beta}(x_0;\omega)+ i \nabla_{x_0} \varphi_\beta(x_0;\omega)\cdot x + O(|x|^2 h)} + R_\eta\\
&= \sum_{n=0}^{M_\eta} \sum_{\beta\in \mathcal{B}_{\mathcal{K},n}} a^0_{\beta}(x_0;\omega) e^{\frac{i}{h} \varphi_{\beta}(x_0;\omega)+ i \nabla_{x_0} \varphi_\beta(x_0;\omega)\cdot x} +R_\eta,
\end{aligned}
\]
where $\|R_\eta\|_{C^0(B(0,h^{-1/3})}\leq \eta$.

Let us write $F_h(x):= \Re(E_h(x,\omega_0)) + \Re(E_h(x, \omega_1))$. Since $\omega \mapsto a^0_{\beta}(x_0;\omega)$ is continuous for every $\beta$, we deduce that there exists $\epsilon_\eta>0$ such that if  
$|\omega_0-\omega_1| \leq \epsilon_\eta$, we have for $x\in B(0,h^{-1/3})$ :
\begin{equation*}
\begin{aligned}
F_h(\psi(hx)) &= \sum_{n=0}^{M_\eta} \sum_{\beta\in \mathcal{B}_{\mathcal{K},n}} \Re\Big{[} a^0_{\beta}(x_0;\omega_0) 
e^{\frac{i}{h} \varphi_{\beta}(x_0;\omega_0)+ 
i \nabla_{x_0} \varphi_\beta(x_0;\omega_0)\cdot x} \\
&+ a^0_{\beta}(x_0;\omega_1) e^{\frac{i}{h} \varphi_{\beta}(x_0;\omega_1)+ i \nabla_{x_0} \varphi_\beta(x_0;\omega_1)\cdot x}\Big{]} + 2R_\eta\\
&= \sum_{n=0}^{M_\eta} \sum_{\beta\in \mathcal{B}_{\mathcal{K},n}} \Re\Big{[}a_{\beta}^0(x_0;\omega_0) \big{(}e^{\frac{i}{h} \varphi_{\beta}(x_0;\omega_0)+ i \nabla_{x_0} \varphi_\beta(x_0;\omega_0)\cdot x} \\
& + e^{\frac{i}{h} \varphi_{\beta}(x_0;\omega_1)+ i \nabla_{x_0} \varphi_\beta(x_0;\omega_1)\cdot x }\big{)} \Big{]} + R'_\eta,
\end{aligned}
\end{equation*}
where $\|R'_\eta\|_{C^0(B(0,h^{-1/3})}\leq 3\eta$. The value of $\eta>0$ will be fixed at the end of the proof.

This can be rewritten in a more condensed way as
\begin{equation}\label{carte}
F_h(\psi(hx))= \sum_{n=0}^{M_\eta} \sum_{\beta\in \mathcal{B}_{\mathcal{K},n}} \Big{[}b_\beta \cos \big{(} x\cdot k_{\beta,\omega_0} + \theta^0_{\beta,h}\big{)} + b_\beta \cos \big{(} x\cdot k_{\beta,\omega_1} + \theta^1_{\beta,h}\big{)}\Big{]} + R'_{\eta}(x).
\end{equation}
Here, we have $k_{\beta,\omega_i} = k_{\beta,\omega_i}(x_0)= \nabla_{x_0} \varphi(x_0,\omega_i)$ for $i=0,1$ and $b_\beta = b_\beta(x_0)= |a_{\beta}(x_0,\omega_0)|$.

\begin{remarque}
If $\mathcal{O}\Subset X$ is an open set, and if $g'$ is a small enough perturbation of $g$ in the sense that $C^k(\mathcal{O})$, then the manifold $(X,g')$ will still satisfy (\ref{Hyperbolique}) and (\ref{pression}), so that the resuls from section \ref{rappel} will apply, and we will have a similar expression for $F_h(\psi(hx))$ on $(X,g')$. Furthermore, all the objects appearing in the decomposition (\ref{carte}) depend on the metric in a continuous way. When we will want to emphasize the dependence of the directions of propagation on the metric, we will write $k_{\beta,\omega_i}(g')$.
\end{remarque}

We want to apply Theorem \ref{phasesrandom} to the function
\begin{equation}\label{utile}
G(x)= G_{x_0}(x):= \sum_{n=0}^{M_\eta} \sum_{\beta\in \mathcal{B}_{\mathcal{K},n}} \Big{[}b_\beta \cos \big{(} x\cdot k_{\beta,\omega_0} + \theta^0_{\beta,h}\big{)} + b_\beta \cos \big{(} x\cdot k_{\beta,\omega_1} + \theta^1_{\beta,h}\big{)}\Big{]}.
\end{equation}

The first hypothesis in Theorem \ref{phasesrandom} is that there are at least six different $k_{\beta,\omega}$ with non-zero amplitudes We know from Remark \ref{Plusieurscoefs} that the $k_{\beta,\omega}$ take different values for different $\beta$. Furthermore, we have by (\ref{nonnul}) that infinitely many amplitudes $b_\beta$ are non-zero. We may therefore find a constant $c_0\bel 0$ and six indices $\beta_i$, $i=1,...,6$ such that $b_{\beta_i}\geq c_0$.

In Theorem \ref{phasesrandom}, the constants $\mathcal{R}_0$ and $\epsilon_3$ depend only on the supremum of the amplitudes, and on the positions and amplitudes associated to the six points mentionned in the statement. In particular, if we can check that the two other hypotheses of Theorem \ref{phasesrandom} are satisfied for metrics $g'$ in a neighbourhood of $g$, then $\mathcal{R}_0$ and $\epsilon_3$ will depend continuously on $g'$.

We may take $N$ large enough, and $h_0$ small enough so that for all $h\leq h_0$ and all $x_0\in \mathcal{O}'$, we have
\begin{equation}
\|R_{h}+R_\epsilon\|_{C^0(B(0,h^{-1/3})} \leq \frac{\epsilon_3}{2}.
\end{equation}

Note that the hypothesis of $\epsilon_1$-non-domination is always satisfied as soon as $\epsilon$ is chosen small enough, since the amplitudes in front of the cosine are two by two equal for close enough directions of propagation.

To make sure that the hypothesis of $\epsilon_0$-independence is satisfied, we must now perturb the metric in a generic way. 

\subsubsection*{Local perturbation of the metric}

The following lemma is standard, and can for example be seen as a consequence of Proposition 5 in \cite{rifford2012closing}. Note that it holds in any dimension. See Figure \ref{deviation} for an illustration of the statement.
\begin{lemme}\label{perturbdevi}
Let $\mathcal{O}\subset X$ be a small open set.
Fix a distance $d_{S^*X}$ on $S^*X$, and a way of computing the $C^k$ distance $d_{C^k(\mathcal{O})}$ between metrics in $\mathcal{G}_{\mathcal{O}}$. 

Let $\rho_1,\rho_2\in S^*X$ such that $\pi_X(\rho_1), \pi_X(\rho_2)\in \partial \mathcal{O}$, $\pi_X(\rho_1)\neq \pi_X(\rho_2)$. We suppose that there exists $T\in \R$ with $\Phi_g^T(\rho_1)=\rho_2$, and $\pi_X\big{(}\Phi_g^t(\rho_1)\big{)}\in \mathcal{O}$ for all $t\in(0, T)$. Then there exists $\epsilon_0\bel 0$ such that the following holds.

Let $\rho'_2\in S^*X$ with $\pi_X(\rho'_2\in \partial \mathcal{O})$ be such that $d_{S^*X}(\rho_2,\rho'_2)=\epsilon\leq \epsilon_0$. Then there exists $g'\in \mathcal{G}_\mathcal{O}$ with $\|g-g'\|_{C^k(\mathcal{O})} = o_{\epsilon\rightarrow 0} (1)$ and $T'\bel 0$ such that $\Phi_{g'}^{T'}(\rho_1)=\rho'_2$ and $\pi_X\big{(}\Phi_{g'}^{t}(\rho_1)\big{)}\in \mathcal{O}$ for all $t\in (0,T')$.
\end{lemme}
\begin{figure}\label{deviation}
    \center
   \includegraphics[scale=0.6]{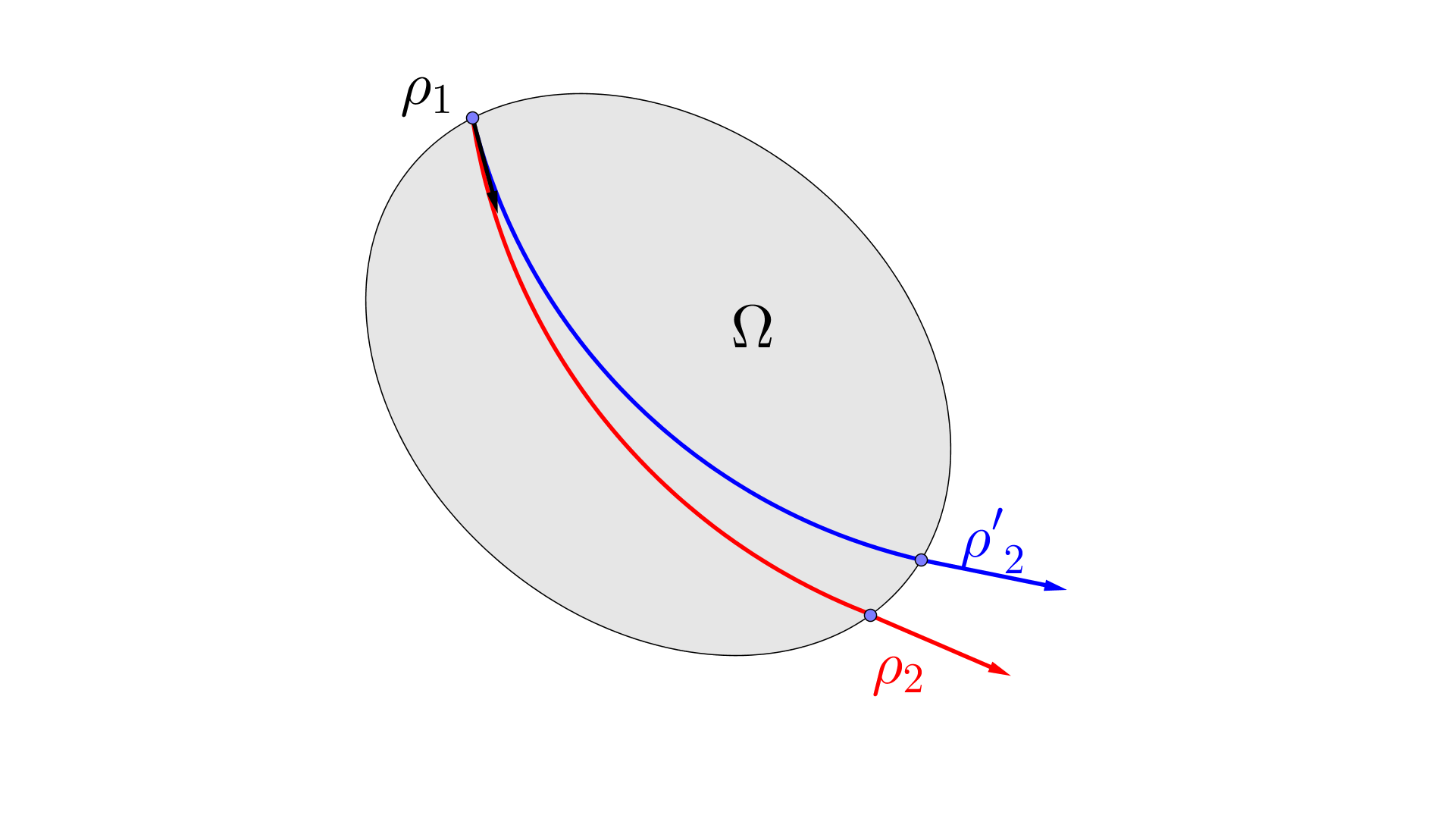}
    \caption{The curve in red going from $\rho_1$ to $\rho_2$ is a geodesic for the metric $g$. By perturbing the metric in $\Omega$, we can obtain a geodesic going from $\rho_1$ to $\rho'_2$.}
    \end{figure}
    \begin{definition}
    We will say that the property $P(g',\epsilon,\omega_0,\omega_1,N)$ is satisfied if the family $\{k_{\beta,\omega_i}(g'); \beta\in \mathcal{B}_{\mathcal{K},n}, n\leq N, i=0,1\}$ is $\epsilon$-independent.
    \end{definition}
  \begin{lemme}  There exists an open set $\mathcal{O}\Subset X_0$ such that for all $\epsilon\bel 0$, $\omega_0,\omega_1\in \mathbb{S}^1$ and $N\in \N$, $P(x,g',\epsilon,\omega_0,\omega_1, N)$ is true for a generic perturbation $g'$ of $g$ in $\mathcal{O}$ in any $C^k(\mathcal{O})$ topology for $k\geq 2$.
\end{lemme}
\begin{proof}

Recall that we write $\Lambda_\omega = \{(x,\omega), x\notin X_0 \}$, and that by (\ref{liendynamique}), $k_{\beta,\omega}(g)$ is the direction of the unique trajectory coming from $\Lambda_\omega$ which is above $x_0$ at time $n$, and which was in $V_{b_k}$ at time $k$ for $k\leq n-1$. 
Therefore, $k_{\beta,\omega}(g)$ depends continuously on $g$ in the $C^k(\mathcal{O})$ topology for $k\geq 2$, and hence $P(x,g',\epsilon,\omega_0,\omega_1, N)$ is true for $g'$ in an open neighbourhood of $g$ by Remark \ref{generik}. Let us show that this open set is dense.

Note that, $X$ being Euclidean near infinity, for any $x_0\in X$ and $\omega\in \Sp^1$, there exists at most one trajectory starting from a point in $\Lambda_\omega$ and going through $x_0$ without going through $X_0$. This is the case precisely when $x_O$ belongs to the Euclidean region, and the trajectory is a straight line. If such a trajectory exists, it therefore corresponds to $\beta=(0,...,0)$.

$x_0$ and $\omega$ being fixed, we may find an open set $\mathcal{O}\Subset X_0$ such that there exists at most one trajectory starting from a point in $\Lambda_\omega$ and going through $x_0$ without going through $\mathcal{O}$. $\mathcal{O}$ being an open set, this property will remain true if we perturb slightly the metric, and if we replace $\omega$ by some $\omega'$ close enough from $\omega$.

For each $\omega_i$, $i=0,1$ and each $\beta\in \mathcal{B}_{k}\neq (0,...,0)$, $k\leq N$, let us take a small open set $\mathcal{O}'_{\beta,\omega_i}\subset \mathcal{O}$ such that
\begin{equation*}
\{t\geq 0~; \pi_X (\Phi^{-t}(x_0,k_{\beta',\omega_j}))\in \mathcal{O}'_{\beta,\omega_i}\}= \left\{
    \begin{split}
   \emptyset \text{ if } \beta' \neq \beta \text{ or } i\neq j\\ 
    ]t_1,t_2[ \text{ with } t_1<t_2 \text{ if } \beta'=\beta \text{ and } i=j.
    \end{split}
  \right.
\end{equation*}
It is always possible to find such open sets, since the trajectories we consider are in finite number, and they are all disjoint. 
 
Let $k$ be a vector close enough from $k_{\beta,\omega_i}$, so that $$\{t\geq 0~; \pi_X (\Phi^{-t}(x_0,k_{\beta,\omega_i}))\in \mathcal{O}'_{\beta,\omega_i}\}= ]t'_1,t'_2[ \neq \emptyset.$$ We have in particular that $\Phi^{t'_1} (x_0,k)$ is close from $\Phi^{t_1} (x_0,k_{\beta,\omega})$.

By Lemma \ref{perturbdevi}, we know that it is possible to perturb the metric in $\mathcal{O}'_{\beta,\omega_i}$ so that the trajectory of $(\Phi_{g'}^t)$ which starts in $(\Phi_g^{-t_2}(x_0,k_{\beta,\omega_i})$ leaves $S^*\mathcal{O}'_{\beta,\omega_i}$ in the future in $\Phi_g^{-t'_1}(x_0,k)$.

By perturbing the metric slightly in such a way, we may therefore modify slightly a direction $k_{\beta,\omega_i}$ as we wish, without changing the other directions  $k_{\beta',j}$.

Since, on the other hand, for $\beta=(0,...,0)$, the family $(k_{\beta,\omega_0},k_{\beta,\omega_1})$ is $\epsilon$-independent as long as we take $\omega_0$ and $\omega_1$ close enough from each other, we deduce from Remark \ref{generik} that the set of metrics $g'$ such that $P(x,g',\epsilon,\omega_0,\omega_1, N)$ is satisfied is dense in a neighbourhood of $g$.
\end{proof}

\subsubsection*{End of the proof of theorem \ref{ivrogne}}
For a $C^k(\mathcal{O})$-generic perturbation of $g$, we may apply Theorem \ref{phasesrandom} to the function $G_{x_0}$ in (\ref{utile}). We obtain that there exists $c>0$ such that for $r$ large enough, this function has at least  $c r^2$ nodal domains which are $\epsilon_3$-stable. By taking $\eta<\epsilon_3/6$, the remainder in (\ref{carte}) can be made smaller that $\epsilon_3/2$, so that $\Re(E_h(\cdot,\omega_0;g')) + \Re(E_h(\cdot, \omega_1;g'))$ has at least $c h^{-2/3}$ nodal domains contained in a ball of radius $h^{2/3}$ around $x_0$.

Since the $k_{\beta,\omega_i}(x_0)$ and $b_{\beta}(x_0)$ depend continuously on $x_0$, we may use Remark \ref{stableperturb} to find $\epsilon_5>0$ small enough, so that for all $x_1\in B(x_0,\epsilon_5)$, $G_{x_1}$ has at least $cr^2/2$ nodal domains which are $\frac{\epsilon_3}{2}$-stable, so that $\Re(E_h(\cdot,\omega_0;g')) + \Re(E_h(\cdot, \omega_1;g'))$ has at least $c h^{-2/3}$ nodal domains contained in a ball of radius $h^{2/3}$ around $x_1$.

We may find $c'h^{-4/3}$ points $x_i\in B(x_0,\epsilon_5)$, with $c'>0$ independent of $h$, such that the balls $B(x_i, h^{2/3})$ are two by two disjoint. By what precedes, in each of these balls, $\Re(E_h(\cdot,\omega_0;g')) + \Re(E_h(\cdot, \omega_1;g'))$ has at least $c h^{-2/3}$ nodal domains. All in all, $\Re(E_h(\cdot,\omega_0;g')) + \Re(E_h(\cdot, \omega_1;g'))$ has at least $c' h^{-2}$ nodal domains in $B(x_0,\epsilon_5)$. This concludes the proof of Theorem \ref{ivrogne}. 
\end{proof}

\bibliographystyle{alpha}
\bibliography{references}

\begin{thebibliography}{BBB03}

\bibitem[BBB03]{berti2003drift}
M.~Berti, L.~Biasco, and P.~Bolle.
\newblock Drift in phase space: a new variational mechanism with optimal
  diffusion time.
\newblock {\em Journal de math{\'e}matiques pures et appliqu{\'e}es},
  82(6):613--664, 2003.

\bibitem[Ber77]{berry1977regular}
M.V. Berry.
\newblock Regular and irregular semiclassical wavefunctions.
\newblock {\em Journal of Physics A: Mathematical and General}, 10(12):2083,
  1977.

\bibitem[BM82]{berard1982inegalites}
P.~B{\'e}rard and D.~Meyer.
\newblock In{\'e}galit{\'e}s isop{\'e}rim{\'e}triques et applications.
\newblock In {\em Annales scientifiques de l'{\'E}cole Normale Sup{\'e}rieure},
  volume~15, pages 513--541, 1982.

\bibitem[Bou14]{bourgain2014toral}
J.~Bourgain.
\newblock On toral eigenfunctions and the random wave model.
\newblock {\em Israel Journal of Mathematics}, 201(2):611--630, 2014.

\bibitem[BW15]{buckley2015number}
J.~Buckley and I.~Wigman.
\newblock On the number of nodal domains of toral eigenfunctions.
\newblock {\em arXiv preprint arXiv:1511.04382}, 2015.

\bibitem[CH67]{CouHil}
R.~Courant and D.~Hilbert.
\newblock {\em Methods of Mathematical Physics, Vol.I}.
\newblock Interscience Publishers Inc. N.Y., 1967.

\bibitem[DZ]{Resonances}
S.~Dyatlov and M.~Zworski.
\newblock {\em Mathematical theory of scattering resonances}.
\newblock Version 0.03, To appear.

\bibitem[HL13]{HanLin}
Q.~Han and F.-H. Lin.
\newblock {\em Nodal Sets of Solutions of Elliptic Differential Equations}.
\newblock Book available on Han's homepage, 2013.

\bibitem[Ing15]{Ing2}
M.~Ingremeau.
\newblock Distorted plane waves on manifolds of nonpositive curvature.
\newblock {\em arXiv preprint arXiv:1512.06818}, 2015.

\bibitem[KH95]{KH}
A.~Katok and B.~Hasselblatt.
\newblock {\em Introduction to the modern theory of dynamical systems}.
\newblock 1995.

\bibitem[KW15]{kurlberg2015non}
P.~Kurlberg and I.~Wigman.
\newblock Non-universality of the {N}azarov--{S}odin constant.
\newblock {\em Comptes Rendus Mathematique}, 353(2):101--104, 2015.

\bibitem[Mel95]{Mel}
R.B. Melrose.
\newblock {\em Geometric {S}cattering {T}heory}.
\newblock Cambridge University Press, 0995.

\bibitem[NS09]{nazarov2009number}
F.~Nazarov and M.~Sodin.
\newblock On the number of nodal domains of random spherical harmonics.
\newblock {\em American Journal of Mathematics}, pages 1337--1357, 2009.

\bibitem[NS15]{nazarov2015asymptotic}
F.~Nazarov and M.~Sodin.
\newblock Asymptotic laws for the spatial distribution and the number of
  connected components of zero sets of gaussian random functions.
\newblock {\em arXiv preprint arXiv:1507.02017}, 2015.

\bibitem[Pee57]{peetre1957generalization}
J.~Peetre.
\newblock A generalization of {C}ourant's nodal domain theorem.
\newblock {\em Mathematica Scandinavica}, pages 15--20, 1957.

\bibitem[Ple56]{pleijel1956remarks}
A.~Pleijel.
\newblock Remarks on {C}ourant's nodal line theorem.
\newblock {\em Communications on Pure and Applied Mathematics}, 9(3):543--550,
  1956.

\bibitem[Rif12]{rifford2012closing}
L.~Rifford.
\newblock Closing geodesics in ${C}^1$ topology.
\newblock {\em J. Differential Geom}, 91(3):361--382, 2012.

\bibitem[Ste25]{stern1925bemerkungen}
A.~Stern.
\newblock {\em Bemerkungen {\"u}ber asymptotisches Verhalten von Eigenwerten
  und Eigenfunktionen}.
\newblock W. Fr. Kaestner, 1925.

\end{thebibliography}
\end{document}